\documentclass[twopage]{scrartcl}

\usepackage[T1]{fontenc}
\usepackage[utf8]{inputenc}
\usepackage[safe]{tipa} % for \textlambda

\expandafter\let\csname equation*\endcsname\relax
\expandafter\let\csname endequation*\endcsname\relax
\usepackage{amsmath}
\usepackage{amsfonts}
\usepackage{amssymb}
\usepackage{mathtools}
\usepackage{mathpartir}
\usepackage[amsmath,thmmarks]{ntheorem}
\usepackage{stmaryrd}
\usepackage{paralist}
\usepackage[pdfborder={0 0 0},naturalnames]{hyperref}
\usepackage{microtype}

\usepackage{mathpazo}

\usepackage{tocloft}
\addtocontents{toc}{\protect\enlargethispage{2cm}}
\tocloftpagestyle{empty}

\allowdisplaybreaks[1]

\title%[The Correctness of Launchbury's Semantics]
      {The Correctness of Launchbury's Natural Semantics for Lazy Evaluation}

 \author%[J. Breitner]
        {Joachim Breitner\\ %\thanks{supported by the Deutsche Telekom Stiftung}\\
	\small{Programming Paradigms Group}\\
        \small{Karlsruhe Institute of Technology, Germany}\\
        \small{\url{breitner@kit.edu}}}

\theorembodyfont{}
\newtheorem{theorem}{Theorem}
\newtheorem{falsetheorem}[theorem]{‘Theorem’}
\newtheorem{counterexample}[theorem]{Counter example}
\newtheorem{lemma}{Lemma}

\theoremstyle{nonumberbreak}
\theoremheaderfont{\itshape}
\theoremsymbol{\ensuremath{\blacksquare}}
\theoremseparator{.}
\newtheorem{proof}{Proof}

\newcommand\pfun{\mathrel{\ooalign{\hfil$\mapstochar\mkern5mu$\hfil\cr$\to$\cr}}}

\newcommand{\sVar}  {\text{\textsf{Var}}}
\newcommand{\sExp}  {\text{\textsf{Exp}}}
\newcommand{\sHeap} {\text{\textsf{Heap}}}
\newcommand{\sVal}  {\text{\textsf{Val}}}
\newcommand{\sValue}{\text{\textsf{Value}}}
\newcommand{\sEnv}  {\text{\textsf{Env}}}

\newcommand{\sFn}[1]{\text{\textsf{Fn}}\,(#1)}
\newcommand{\sFnProj}[2]{#1\,\downarrow_{\text{\textsf{Fn}}}\,#2}

\newcommand{\req}[1]{=_{#1}}

\newcommand{\sApp}[2]{#1\;#2}
\newcommand{\sLam}[2]{\text{\textlambda} #1.\, #2}
\newcommand{\sLet}[2]{\text{\textsf{let}}\ #1\ \text{\textsf{in}}\ #2}
\newcommand{\sred}[4]{#1 : #2 \Downarrow #3 : #4}
\newcommand{\ssred}[4]{#1 : #2 \mathrel{\rotatebox[origin=c]{90}{$\Lleftarrow$}} #3 : #4}
\newcommand{\sRule}[1]{\text{{\textsc{#1}}}}

\newcommand{\dom}[1]{\text{dom}\;#1}

\newcommand{\xeng}{x_1 = e_1, \ldots, x_n = e_n}
\newcommand{\xen}{x_1\mapsto e_1, \ldots, x_n\mapsto e_n}

\newcommand{\dsem}[2]{\llbracket #1 \rrbracket_{#2}}
\newcommand{\esem}[1]{\{\!\!\!\{#1\}\!\!\!\}}
\newcommand{\dsemu}[2]{\llbracket #1 \rrbracket^{\text{u}}_{#2}}
\newcommand{\esemu}[1]{\{\!\!\!\{#1\}\!\!\!\}^{\text{u}}}

\newcommand{\case}[1]{\par\smallskip\noindent\textbf{Case:} #1\nopagebreak\par\noindent\ignorespaces}
\newcommand{\beginright}{&\mathrel{\phantom{=}}}
\newcommand{\aexpl}[1]{&\mathrel{\phantom{=}}\left\{\text{ #1 }\right\}}

\newcommand{\shortcite}{\cite}

\begin{document}

\label{firstpage}

\maketitle

\begin{abstract}
In his seminal paper „A Natural Semantics for Lazy Evaluation“, John Launchbury
proves his semantics correct with respect to a denotational semantics. We
machine-checked the proof and found it to fail, and provide two ways to fix it:
One by taking a detour via a modified natural semantics with an explicit stack,
and one by adjusting the denotational semantics of heaps.
\end{abstract}

\tableofcontents

\section{Introduction}

The Natural Semantics for Lazy Evaluation created by Launchbury \shortcite{launchbury} has turned out to be a popular and successful foundation for theoretical treatment of lazy evaluation, especially as the basis of semantic extensions \cite{nakata, nakata_blackhole, distributed, mixed, parallel}. Therefore, its correctness and adequacy is important in this field of research. The original paper defines a standard denotational semantics to prove the natural semantics correct against, and outlines the adequacy proof.

Unfortunately, the correctness proof is flawed: To show that a closed term evaluates to a value with the same denotation by induction on the derivation of the natural semantics, Launchbury generalizes the correctness statement to non-empty semantic environments. This is Theorem~2 in \cite{launchbury}, and a counter-example can be given (see Section~\ref{counterexample}). Several later works based on Launchbury rely on this proof and hence also contain the error.

Fortunately, the error only affects the proof and the correctness theorem still holds for empty environments. One way to prove this is to add an explicit stack to the judgments of the semantics, to capture more of the context of evaluation (Section \ref{stackedsemantics}). This way, we need not generalize the statement to arbitrary environments for the inductive proof. As both semantics are equivalent, this provides the correctness of Launchbury's semantics.

Another way to fix the problem is to modify the meaning of the $\sqcup$ operator in the semantics of heap: If this is understood to be a right-sided update instead of a least upper bound, the original proof by Launchbury goes through almost unmodified. We reproduce this in Section \ref{updsemantics}.

All definitions, propositions and proofs were mechanically verified using the theorem prover Isabelle/HOL \cite{isabelle} and can be found in the Archive of Formal Proofs \cite{afp}. Therefore, we take the liberty to concentrate on the important steps and tricky calculations of the proofs here. In particular, we will not explicate the treatment of names, we implicitly expect heaps to be distinctly named and we do not show that partial operations like $\sqcup$ are defined where used. For all gory details, we refer the interested reader to the Isabelle proof document, which not only contains the full proofs, but also has the LaTeX code of the theorems automatically generated from the proved statements.

Our contributions are:
\begin{compactitem}
\item We exhibit an error in the original correctness proof of Launchbury's semantics.
\item We provide a variant of Launchbury's semantics that allows for simpler proofs.
\item We prove its correctness and, by equivalency, the correctness of the original semantics.
\item We show that a modification to the original denotational semantics allows the original correctness proof to go through.
\item All these results are formally proven and machine-checked.
\end{compactitem}

\section{Launchbury's semantics}

Launchbury defines a semantics for a simple untyped lambda calculus consisting of variables, lambda abstraction, applications and recursive let bindings:
\begin{alignat*}{2}
x,y,z,w &\in \sVar
\displaybreak[1]
\\
e &\in
\sExp &&\Coloneqq
\sLam x e
\mid \sApp e x
\mid x \mid
\sLet {\xeng} e
\end{alignat*}

It is worth noting that the term on the right hand side of an application has to be a variable. A general lambda term of the form $\sApp{e_1}{e_2}$ would have to be pre-processed to $\sLet{x = e_2}{\sApp{e_1}x}$ before it can be handled by this semantics.

\subsection{Natural semantics}

\begin{figure}
\begin{mathpar}
\inferrule
{ }
{\sred{\Gamma}{\sLam xe}{\Gamma}{\sLam xe}}
\sRule{Lam}
\and
\inferrule
{\sred{\Gamma}e{\Delta}{\sLam y e'}\\ \sred{\Delta}{e'[x/y]}{\Theta}{v}}
{\sred\Gamma{\sApp e x}\Theta v}
\sRule{App}
\and
\inferrule
{\sred\Gamma e \Delta v}
{\sred{\Gamma, x\mapsto e} x {\Delta, x\mapsto v}{v}}
\sRule{Var}
\and
\inferrule
{\sred{\Gamma,\xen} e \Delta v}
{\sred{\Gamma}{\sLet{\xeng}e} \Delta v}
\sRule{Let}
\end{mathpar}
\caption{The original natural semantics}
\label{fig:natsem}
\end{figure}

Launchbury gives this language meaning by a natural semantics, specified with the rules in Figure \ref{fig:natsem}, which obey the following naming convention for heaps and values:
\begin{alignat*}{2}
\Gamma, \Delta, \Theta &\in \sHeap &&= \sVar \pfun \sExp
\displaybreak[1]\\
v &\in \sVal &&\Coloneqq \sLam x e
\end{alignat*}

A heap is a partial function from variables to expressions. The domain of an heap $\Gamma$, written $\dom\Gamma$, is the set of variables bound in the heap.

A value is an expression in weak head normal form. A judgment of the form $\sred \Gamma e \Delta v$ means that the expression $e$ with the heap $\Gamma$ reduces to $v$, while modifying the heap to $\Delta$.

In this work we treat naming and binding naively, as these issues are not essential for our results. For example, it is understood that the variables in the premise of the rule \sRule{Let} are actually fresh and distinct from the variables bound in the let expression. Likewise the substitution $e'[x/y]$ replaces all free occurrences of $y$ in $e'$ with $x$  and a “fresh variable” is fresh with regard to the derivation tree and not just the elements of the current judgment. To express the latter issue rigorously, we would have to add to the judgment a set of variables to avoid, following Sestoft \shortcite{sestoft}.

We can take these shortcuts here with good conscience, as all proofs also exist in machine-checked form and there, naming has been handled rigorously.

\subsection{Denotational semantics}

In order to show that the natural semantics behaves as expected, Launchbury defines a standard denotational semantics for expressions and heaps, following Abramsky \shortcite{abramsky}. The semantic domain $\sValue$ is the initial solution to the domain equation 
\[
\sValue = (\sValue \to \sValue)_\bot,
\]
which allows to distinguish $\bot$ from $\lambda x. \bot$. Lifting between $\sValue \to \sValue$ and $\sValue$ is performed using the injection $\sFn \_$ and projection $\sFnProj{\_}{}$. Values are partially ordered by~$\sqsubseteq$.

A semantic environment maps variables to values
\begin{alignat*}{2}
\rho \in \sEnv &= \sVar \to \sValue
\end{alignat*}
and the initial environment $\rho_\bot$ maps all variables to $\bot$.

The domain of an environment $\rho$, written $\dom\rho$, is the set of variables that are not mapped to $\bot$.
The environment $\rho|_S$, where $S$ is a set of variables, is the domain-restriction of $\rho$ to $S$:
\[
(\rho |_S)\, x = 
\begin{cases}
\rho\, x,& \text{if } x \in S\\
\bot& \text{if } x \not\in S.
\end{cases}
\]
The environment $\rho\setminus S$ is defined as the the domain-restriction of $\rho$ to the complement of $S$, i.e.\ \mbox{$\rho\setminus S \coloneqq \rho|_{\sVar \setminus S}$}.

The semantics of expressions and heaps are mutually recursive. The meaning of a expression $e \in \sExp$ in an environment $\rho \in \sEnv$ is written as $\dsem e \rho \in \sValue$ and is defined as
\begin{align*}
\dsem{\sLam x e}\rho &\coloneqq \sFn{\lambda v. \dsem e {\rho \sqcup \{x \mapsto v\}}}\\
\dsem{\sApp e  x}\rho &\coloneqq \sFnProj {\dsem e \rho}{\dsem x \rho}\\
\dsem{x}\rho &\coloneqq \rho\,x\\
\dsem{\sLet{\xeng}e}\rho &\coloneqq \dsem e {\esem{\xeng}\rho.}
\end{align*}

The meaning of a heap $\Gamma \in \sHeap$ in an environment $\rho$ is $\esem \Gamma \rho \in \sEnv$, defined as
\begin{align*}
\esem{ \xen}\rho 
&= \mu \rho'. \rho \sqcup (x_1 \mapsto \dsem{e_1}{\rho'}, \ldots, x_n \mapsto \dsem{e_n}{\rho'})
\end{align*}
where $\sqcup$ is the least upper bound and $\mu$ is the least-fixed-point operator. This definition only makes sense when all occurring least upper bounds actually exist. This is the case here, as shown in the machine-checked proofs.

We sometimes write $\esem{\Gamma}$ instead of $\esem{\Gamma}{\rho_\bot}$, and we write $\dsem{\xen}\rho$ for $(x_1\mapsto \dsem{e_1}\rho,\ldots,x_n\mapsto \dsem{e_n}\rho)$ as it occurs in the definition of $\esem{\_}\_$. In an expression $\esem{\Gamma}(\esem{\Delta}\rho)$ we omit the parentheses and write $\esem{\Gamma}\esem{\Delta}\rho$.

The relation $\sqsubseteq$ on environments is $\sqsubseteq$ on $\sValue$, lifted pointwise.

Launchbury additionally introduces the partial order $\le$ on environments, where $\rho \le \rho'$ is defined as $\forall x. \rho\, x \ne \bot \implies \rho\,x =\rho'\,x$. While this captures, as intended, the concept of $\rho'$ adding bindings to $\rho$, the use of this definition in the proofs is problematic, as discussed in the following section.

The notation $\esem{\Gamma}\rho \preceq \esem{\Delta}{\rho'}$ is an abbreviation for the expression $\dom\Gamma \subseteq \dom\Delta \wedge \forall x\in \dom\Gamma. (\esem{\Gamma}\rho)\,x = (\esem{\Delta}{\rho'})\,x$, and $\preceq$ is transitive.

\subsection{The correctness theorem and the counter example}
\label{counterexample}

The main correctness theorem for the natural semantics is
\begin{theorem}
If $\sred \Gamma e \Delta v$, then $\dsem{e}{\esem{\Gamma}} = \dsem{v}{\esem{\Delta}}$.
\label{thm:main}
\end{theorem}

Launchbury generalizes this to arbitrary environments, supposedly to obtain a stronger result, and to enable a proof by induction:

\begin{falsetheorem}
If $\sred \Gamma e \Delta v$ holds, then for all environments $\rho \in \sEnv$, $\dsem{e}{\esem{\Gamma}{\rho}} = \dsem{v}{\esem{\Delta}{\rho}}$ and $\esem\Gamma\rho \le \esem\Delta\rho$.%
\label{thm:false}
\end{falsetheorem}

\begin{counterexample}
`Theorem' \ref{thm:false} does not hold for $e = x$, $v = \sLam{a}{\sLet{b = b}b}$, $\Gamma = \Delta = (x \mapsto v)$ and $\rho = (x \mapsto \sFn{\lambda \_. \sFn{\lambda x.x}})$. 
\end{counterexample}

\begin{proof}
Note that the denotation of $v$ is $\sFn{\lambda \_. \bot}$ in every environment. We have $\sred\Gamma e\Delta v$, so according to the theorem, $\dsem{e}{\esem{\Gamma}{\rho}} = \dsem{v}{\esem{\Delta}{\rho}}$ should hold, but
\begin{align*}
\dsem{e}{\esem\Gamma\rho}
&= \big(\esem \Gamma\rho\big)\,x \\
&= \rho\,x \sqcup \dsem{v}{\esem{\Gamma}\rho}\\
&= \sFn{\lambda \_. \sFn{\lambda x.x}} \sqcup \sFn{\lambda \_. \bot} \\
&= \sFn{\lambda \_. \sFn{\lambda x.x} \sqcup \bot} \\
&= \sFn{\lambda \_. \sFn{\lambda x.x}} \\
& \ne \sFn{\lambda \_. \bot} \\
&= \dsem{v}{\esem{\Delta}{\rho}}.
\end{align*}
\end{proof}

Tracing the counter example through the original proof we find that in the case for $\sRule{Var}$, the equation $({\esem{\Delta, x\mapsto z}\rho})\, x = \dsem{z}{\esem{\Delta, x\mapsto z}\rho}$ is used, while in fact $({\esem{\Delta, x\mapsto z}\rho})\, x = \rho\, x \sqcup \dsem{z}{\esem{\Delta, x\mapsto z}\rho}$ holds. So the problem occurs when $\rho$ contains bindings that are, in some way, incompatible with the semantics of $\Gamma$.

\subsubsection{Failed attempts at fixing the proof}

The main Theorem \ref{thm:main} is not affected by the flaw, as such “bad” environments do not occur during the evaluation of closed expressions. So it seems that `Theorem' \ref{thm:false} can be fixed by restricting $\rho$ to a certain subset of all environments that somehow comprises of all environments occurring in the inductive proof. Such a property will relate $\rho$ with the semantics of $\Gamma$, has to hold for $\rho_\bot$ and needs to be strong enough to hold for the inductive cases. In particular, it needs to be preserved by evaluation. Unfortunately, the required property does not appear to have a simple definition. As negative results can be very educating as well, we briefly discuss some apparent definitions and why they fail.

As we have to relate $\rho$ with the semantics of $\Gamma$, a first attempt is to restrict `Theorem' \ref{thm:false} to environments for which $\forall x\mapsto e \in \Gamma.\,\rho\,x \sqsubseteq \dsem e \rho$ holds. But this property is not preserved by evaluation: Evaluating $x$ in the heap $\Gamma = (x \mapsto \sLet{y = z}{\sLam \_ y})$  yields the updated heap $\Delta = (x \mapsto \sLam{\_} y, y \mapsto z)$. An environment with $\rho\, x = \sFn{\lambda \_. \sFn{\lambda \_.\bot}}$, $\rho\, y = \bot$ and $\rho z = \sFn{\lambda \_.\bot}$ fulfills the property with respect to $\Gamma$, but not $\Delta$.

Obviously it is not sufficient to relate $\rho$ with the entries of the heap individually. Therefore, the next attempt is to consider environments for which $\forall x\mapsto e \in \Gamma.\,\rho\,x \sqsubseteq \dsem e {\esem{\Gamma}\rho}$ or, equivalently $\forall x\mapsto e \in \Gamma.\,(\esem{\Gamma}\rho)\,x \sqsubseteq \dsem e {\esem{\Gamma}\rho}$ holds. Here, our counter example is evaluating $x$ in $\Gamma=(x\mapsto \sLam zz, y \mapsto x)$. In the inductive case of rule $\sRule{Var}$, the heap becomes $(y \mapsto x)$, so an environment $\rho$ with $\rho\, x= \bot$ and $\rho\, y = \sFn{\lambda z.z}$ fulfills the condition with regard to $\Gamma$, but not with regard to $(y\mapsto x)$, so the condition is too weak to allow for an inductive proof.

A different approach would be to demand that the domain of $\rho$ is distinct from the set of variables bound in $\Gamma$ and $\Delta$. But in the case for \sRule{Var} in Launchbury's proof
the induction hypothesis is invoked for a $\rho'$ bound by the least-fixed-point operator in the term $(\mu \rho'. \esem{\Gamma}{\rho'} \sqcup (x \mapsto \dsem{e}{\esem{\Gamma}{\rho'}}) \sqcup \rho)$ and clearly the domain of $\rho'$ will include the variables bound by $\Gamma$, so again this requirement is too strong.

As there seems to be no easy characterization of the environments $\rho$ for which we need the result of `Theorem' \ref{thm:false}, we had to find a different proof, which is provided in Section~\ref{stackedsemantics}.

\subsubsection{A suitable alternative to $\le$}
\label{lealternatives}

A second pitfall is the use of the partial order $\le$ to capture that the denotation of the heap is not modified where defined, but only extended with new bindings. Again in the proof of case \sRule{Var} Launchbury first shows that $\esem{\Gamma, x\mapsto e}\rho \le \esem{\Delta, x \mapsto z}\rho$ and from that concludes $(\esem{\Gamma, x\mapsto e}\rho)\,x = (\esem{\Delta, x \mapsto z}\rho)\,x$ in order to show the first part of the correctness statement. But for that conclusion one would first have to show that $(\esem{\Gamma, x\mapsto e}\rho)\,x \ne \bot$, which is not true in general.

One approach to fix this would be to model environments as partial maps from $\sVar$ to $\sValue$. Then we could differentiate between variables not bound in $\rho$ and variables bound to $\bot$ and have $\le$ state that bound variables have to be equal, whether they are $\bot$ or not. This is the approach taken in \cite{afp}, where it also helps with other technical issues of the machine formalization, but it adds notational complexity that is unwanted for this presentation.

Similarly, one could define a relation $\rho \req{S} \rho'$ which is defined as $\forall x \in S. \rho\,x = \rho'\,x$, and always state the set of variables to be compared. But again the notational overhead is considerable.

Therefore, in this work, we simply use $\esem{\Gamma}\rho \preceq \esem{\Delta}{\rho'}$ to express that $\dom\Gamma \subseteq \dom\Delta$ and $\forall x\in \dom\Gamma. (\esem{\Gamma}\rho)\,x = (\esem{\Delta}{\rho'})\,x$.

\section{The stacked semantics}
\label{stackedsemantics}

In order to prove Theorem \ref{thm:main}, we take a detour via a different semantics that allows to perform the induction without the problematic generalization. The judgments of this semantics are of the form $\ssred \Gamma {\Gamma'} \Delta {\Delta'}$ and have new fields (storing a list of variable-expression pairs) that not only contain the expression currently under evaluation, but also the expressions whose evaluation has caused the current evaluation, together with their respective names. Such expressions are always either variables or applications, so the resulting data structures strongly resembles an evaluation stack, consisting of update frames and function parameters. We re-use the syntax of heaps (which are unordered) here, but keep in mind that $\Gamma'$ and $\Delta'$ are ordered, so that we can talk about the topmost expression.

The rules are given in Figure \ref{fig:stacksem}. When compared to Figure \ref{fig:natsem} one will find that $\sred{\Gamma}e{\Delta}v$ has become $\ssred{\Gamma}{z \mapsto e, \Gamma'}{\Delta}{z \mapsto v, \Delta'}$, i.e. the expression under evaluation is the topmost expression in the stack. The variables bound in the heap and stack on either side of the rule are always distinct -- another detail that is not further discussed here, but handled in the formal development.

The rules \sRule{Lam} and \sRule{Let} correspond directly to their counterpart in the original semantics.

The rule \sRule{App} takes the application $\sApp e x$ apart. First, the expression $e$ is bound to a new name $w$ and put on top of the stack, where $\sApp e x$ is replaced by $\sApp w x$, in order to evaluate $e$. This evaluation provides a lambda abstraction $\sLam y {e'}$ The body thereof is then, with $x$ substituted for $y$, continued to be evaluated.

The rule \sRule{Var} just shuffles between the heap and the stack: In order to evaluate a variable, its binding is removed from the heap and put on top of the stack. After evaluation, the binding is updated with the value and moved back to the heap.

Note that it follows from the rules that if we have $\ssred \Gamma {z \mapsto e, \Gamma'} {\Delta} {w \mapsto v, \Delta'}$, then $w = z$ and $\Delta' = \Gamma'$, i.e.~during one step of evaluation, only the topmost expression on the stack can change. We deliberately keep this redundancy in the semantics for a more natural presentation of, for example, the correctness statement and to allow for later extensions that might want to modify the stack, such as garbage collection.

%A simple induction shows
%\begin{lemma}
%\label{heapdomsubset}
%$\ssred \Gamma {\Gamma'} {\Delta} {\Delta'} \implies \dom{(\Gamma,\Gamma')} \subseteq \dom{(\Delta,\Delta')}$
%\end{lemma}

\begin{figure}
\begin{mathpar}
\inferrule
{ }
{\ssred{\Gamma}{z \mapsto \sLam xe, \Gamma'}{\Gamma}{z \mapsto \sLam xe, \Gamma'}}
\sRule{Lam}
\and
\inferrule
{\ssred{\Gamma}{w \mapsto e, z \mapsto \sApp w x, \Gamma'}{\Delta}{w \mapsto \sLam y e', z \mapsto \sApp w x, \Delta'}\\
 \ssred{\Delta}{z \mapsto e'[x/y], \Delta'}{\Theta}{\Theta'}}
{\ssred\Gamma{z \mapsto \sApp e x, \Gamma'}\Theta{\Theta'}}
\sRule{App}
\and
\inferrule
{\ssred\Gamma {x \mapsto e, z\mapsto x, \Gamma'} \Delta {x \mapsto v, z\mapsto x, \Delta'}}
{\ssred{\Gamma, x\mapsto e}{z \mapsto x, \Gamma'}{\Delta, x \mapsto v}{z \mapsto v, \Delta'}}
\sRule{Var}
\and
\inferrule
{\ssred{\Gamma,\xen} {z \mapsto e, \Gamma'}\Delta {\Delta'}}
{\ssred{\Gamma}{z \mapsto \sLet{\xeng}e, \Gamma'} \Delta{\Delta'}}
\sRule{Let}
\end{mathpar}
\caption{The stacked semantics}
\label{fig:stacksem}
\end{figure}

\subsection{Equivalency with the natural semantics}

The stacked semantics is equivalent to the original semantics in the following sense:

\begin{theorem}
For all $\Gamma, \Gamma', \Delta, e, v$ we have
\label{thm:equiv}
\[
(\exists\,z~\Delta'.~ \ssred \Gamma {z \mapsto e, \Gamma'} {\Delta} {z \mapsto v, \Delta'})
\iff
\sred{\Gamma}e\Delta v.
\]
\end{theorem}

\begin{proof}
Both directions are proved by induction.
\end{proof}

\subsection{Correctness}

For the stacked semantics, we prove correctness with respect to the denotational semantics, in the sense that reduction of a heap and stack preserves their denotation:

\begin{theorem}
For all $\Gamma, \Gamma', \Delta, \Delta'$ we have 
\label{thm:stackedcorrectness}
\[
\ssred \Gamma {\Gamma'} {\Delta} {\Delta'} \implies \esem{\Gamma, \Gamma'} \preceq \esem{\Delta, \Delta'}.
\]
\end{theorem}

As this proof is one of the main contributions of this paper, we spell it out in greater detail, and use the technical lemmas found in appendix \ref{sec:denprops}.

\begin{proof}
by induction on the derivation of $\ssred \Gamma {\Gamma'} {\Delta} {\Delta'}$.

\case{\sRule{Lam}}
We need to show $\esem{\Gamma, \Gamma'} \preceq \esem{\Gamma, \Gamma'}$, which holds trivially.

\case{\sRule{App}}
We have
\[
\ssred{\Gamma}{w \mapsto e, z \mapsto \sApp w x, \Gamma'}{\Delta}{w \mapsto \sLam y e', z \mapsto \sApp w x, \Delta'}
\text{ and }
\ssred{\Delta}{z \mapsto e'[x/y], \Delta'}{\Theta}{\Theta'},
\]
so by the induction hypothesis,
\[
\esem{\Gamma, w \mapsto e, z \mapsto \sApp w x, \Gamma'} \preceq \esem{\Delta, w \mapsto \sLam y e', z \mapsto \sApp w x, \Delta'}
\]
and
\[
\esem{\Delta,z \mapsto e'[x/y], \Delta'} \preceq \esem{\Theta,\Theta'}.
\]
We need to show $\esem{\Gamma, z \mapsto \sApp e x, \Gamma'} \preceq \esem{\Theta,\Theta'}$:
\begin{align*}
\esem{\Gamma, z \mapsto \sApp e x, \Gamma'} &= \esem{\Gamma, w \mapsto e, z\mapsto \sApp e x, \Gamma'} \setminus \{w\} \\
\aexpl{adding a fresh variable, Lemma \ref{lem:addvar}} \\
&= \esem{\Gamma, w \mapsto e, z \mapsto \sApp w x, \Gamma'} \setminus \{w\} \\
\aexpl{substituting the indirection, Lemma \ref{lem:exp_var_subst}} \\
&\preceq \esem{\Delta, w \mapsto \sLam y e', z \mapsto \sApp w x, \Delta'} \setminus \{w\}\\
\aexpl{by the induction hypothesis}  \displaybreak[1]\\
&= \esem{\Delta, w \mapsto \sLam y e', z \mapsto \sApp (\sLam y e') x, \Delta'} \setminus \{w\} \\
\aexpl{substituting the indirection again, Lemma \ref{lem:exp_var_subst}} \displaybreak[1]\\
&= \esem{\Delta, z \mapsto \sApp (\sLam y e') x, \Delta'}\\
\aexpl{removing the fresh variable again, Lemma \ref{lem:addvar}} \displaybreak[1]\\
&= \esem{\Delta, z \mapsto e'[x/y], \Delta'} \\
\aexpl{semantics of application} \\
&\preceq \esem{\Theta,\Theta'}\\
\aexpl{by the induction hypothesis.}
\end{align*}

\case{\sRule{Var}}
We have 
\[
\esem{\Gamma, x \mapsto e, z\mapsto x, \Gamma'} \preceq \esem{\Delta, x \mapsto v, z\mapsto x, \Gamma'}
\]
by the induction hypothesis. Using Lemma \ref{lem:var_var_subst} on the right hand side, we obtain
\[
\esem{\Gamma, x\mapsto e,z \mapsto x, \Gamma'} \preceq \esem{\Delta, x \mapsto v, z \mapsto v, \Gamma'}.
\]

\case{\sRule{Let}}
We have
\begin{align*}
\beginright
\esem{\Gamma,z \mapsto \sLet{\xeng}e, \Gamma'}\\
&\preceq \esem{\Gamma,\xen, z \mapsto e, \Gamma'} \\
\aexpl{by unfolding the let-expression, Lemma \ref{lem:let_unfold}} \\
&\preceq \esem{\Delta,\Delta'}\\
\aexpl{by the induction hypothesis}
\end{align*}
\end{proof}

From the correctness of the stacked semantics we can easily obtain the correctness of the original semantics:

\begin{proof}[of Theorem \ref{thm:main}]
From $\sred \Gamma e \Delta v$, we have $\ssred{\Gamma}{z \mapsto e}{\Delta}{z\mapsto v}$ for a fresh $z$ by Theorem \ref{thm:equiv}. By the theorem just shown, we have $\esem{\Gamma, z \mapsto e} \preceq \esem{\Delta, z\mapsto v}$. This implies $\dsem{e}{\esem{\Gamma, z \mapsto e}} = \dsem{v}{\esem{\Delta, z\mapsto v}}$ by Lemma \ref{lem:esem_this}. As $z$ is fresh, by Lemma \ref{lem:see_through_fresh} we have $\dsem{e}{\esem{\Gamma}} = \dsem{v}{\esem{\Delta}}$.
\end{proof}

\section{The update-based semantics}
\label{updsemantics}

We have found another way to fix Launchbury's correctness proof: We modify the denotational semantics of heaps to be
\begin{align*}
\esemu{ \Gamma }\rho &= \mu \rho'. \rho + \dsemu{\Gamma}{\rho'},
\end{align*}
where
\[
(\rho + \dsemu{\Gamma}{\rho'})\, x = 
\begin{cases}
(\dsemu{\Gamma}{\rho'})\, x, &\text{if } x \in \dom \Gamma\\
\rho\, x &\text{otherwise},
\end{cases}
\]
i.e. we replace the least upper bound operator by a right-sided update, and otherwise let $\dsemu{\_}\_$ be defined by the same equations as $\dsem{\_}\_$.

Interestingly, the denotational semantics of expressions are the same under both definitions:
\begin{lemma}
For all $e$ and $\rho$, $\dsem{e}{\rho} = \dsemu{e}{\rho}$.
\label{lem:deneq}
\end{lemma}

\begin{proof}
by induction on $e$. The interesting case is \sRule{Let}, where we use that $\esem{\Gamma}{\rho} = \esemu{\Gamma}{\rho}$ if the domains of $\Gamma$ and $\rho$ are distinct, which is the case as the variables introduced on the heap in rule \sRule{Let} are fresh.
\end{proof}

In the following we will only mention $\dsem{\_}\_$.

\subsection{Correctness}

Using the modified denotational semantics we can state `Theorem' \ref{thm:false} as a theorem and prove it:
\begin{theorem}
If $\sred \Gamma e \Delta v$ holds, then for all environments $\rho \in \sEnv$, $\dsem{e}{\esemu{\Gamma}{\rho}} = \dsem{v}{\esemu{\Delta}{\rho}}$ and $\esemu\Gamma\rho \preceq \esemu\Delta\rho$.%
\label{thm:thm2}
\end{theorem}

Our proof follows Launchbury's steps quite closely, the only differences are the use of $\preceq$ instead of $\le$ and the slightly different iterative fixed-point expression in case \sRule{Var}. Nevertheless we reproduce it here for completeness and clarifying details. The required technical lemmas about the denotational semantics are compiled in appendix \ref{updsemanticsprops}.

\begin{proof}
This is essentially the proof in \cite{launchbury}, which proceeds by induction on the derivation of $\sred \Gamma e \Delta v$.

\case{\sRule{Lam}}
This case is trivial.

\case{\sRule{App}}
By the induction hypothesis we know
$\dsem{e}{\esemu{\Gamma}\rho} = \dsem{\sLam y {e'}}{\esemu{\Delta}\rho}$ and $\esemu{\Gamma}\rho \preceq \esemu{\Delta}\rho$ as well as $\dsem{e'[x/y]}{\esemu{\Delta}\rho} = \dsem{v}{\esemu{\Theta}\rho}$ and $\esemu{\Delta}\rho \preceq \esemu{\Theta}\rho$.

While the second part follows from the transitivity of $\preceq$, the first part is a simple calculation:
\begin{align*}
\dsem{\sApp{e}{x}}{\esemu{\Gamma}\rho} &= \sFnProj{\dsem{e}{\esemu{\Gamma}\rho}}{\dsem{x}{\esemu{\Gamma}{\rho}}}\\
\aexpl{by the denotation of application} \\
&= \sFnProj{\dsem{\sLam y {e'}}{\esemu{\Delta}\rho}}{\dsem{x}{\esemu{\Gamma}{\rho}}}\\ 
\aexpl{by the induction hypothesis} \\
&= \sFnProj{\dsem{\sLam y {e'}}{\esemu{\Delta}\rho}}{\dsem{x}{\esemu{\Delta}{\rho}}}\\ 
\aexpl{by the induction hypothesis and the definition of $\preceq$} \\
&= \dsem{e'}{(\esemu{\Delta}\rho)(y \mapsto \dsem{x}{\esemu{\Delta}{\rho}})}\\ 
\aexpl{by the denotation of lambda abstraction} \\
&= \dsem{e'[x/y]}{\esemu{\Delta}\rho}\\ 
\aexpl{by substitution Lemma \ref{lem:subst}} \\
&= \dsem{v}{\esemu{\Theta}\rho}\\
\aexpl{by the induction hypothesis}
\end{align*}

\case{\sRule{Var}}
We know that $\dsem{e}{\esem{\Gamma}\rho'}=\dsem{v}{\esem{\Delta}\rho'}$ and $\esem{\Gamma}\rho' \preceq \esem{\Delta}\rho'$ for all $\rho'\in\sEnv$.

We begin with the second part:
\begin{align*}
\esemu{x \mapsto e,\Gamma}{\rho} &= \mu \rho'.\,\rho + (\esemu{\Gamma}{\rho'})|_{\dom\Gamma} + (x \mapsto \dsem{e}{\esemu{\Gamma}\rho'}) \\
\aexpl{by Lemma \ref{lem:iter}} \\
&= \mu \rho'.\, \rho + (\esemu{\Gamma}{\rho'})|_{\dom\Gamma} + (x \mapsto \dsem{v}{\esemu{\Delta}\rho'}) \\
%\aexpl{by the induction hypothesis, invoked for $\rho'$!} \\
\aexpl{\parbox{\widthof{by the induction hypothesis. Note that}}{\raggedright by the induction hypothesis. Note that we invoke it for $\rho'$ with $\rho' \ne \rho$!}} \\
&\preceq \mu \rho'.\, \rho + (\esemu{\Delta}{\rho'})|_{\dom\Delta} + (x \mapsto \dsem{v}{\esemu{\Delta}\rho'}) \\
\aexpl{by induction and the monotonicity of $\mu$ with regard to $\preceq$} \\
&= \esemu{x \mapsto v, \Delta}{\rho}\\
\aexpl{by Lemma \ref{lem:iter}}
\end{align*}

The first part now follows from the second part:
\begin{align*}
\dsem{x}{\esemu{x \mapsto e, \Gamma}\rho} &=
(\esemu{x \mapsto e, \Gamma}\rho)\,x &
\aexpl{by the denotation of variables}\\
&= (\esemu{x \mapsto v, \Delta}\rho)\,x &
\aexpl{by the first part and $x\in\dom(x\mapsto e, \Gamma)$}\\
&=\dsem{x}{\esemu{x \mapsto v, \Delta}\rho}&
\aexpl{by the denotation of variables}\\
&= \dsem{v}{\esemu{x \mapsto v, \Delta}\rho}.
\end{align*}

\case{\sRule{Let}}
We know that $\dsem{e}{\esem{\Gamma, \xen}\rho} = \dsem{v}{\esem{\Delta}{\rho}}$ and $\esem{\Gamma, \xen}\rho \preceq \esem{\Delta}{\rho}$.
For the first part we have
\begin{align*}
\beginright
\dsem{\sLet{\xeng}e}{\esem{\Gamma}\rho} \\
&= \dsem{e}{\esem{\xen}{\esem{\Gamma}\rho}} &
\aexpl{by the denotation of let-expressions} \\
&= \dsem{e}{\esem{\Gamma, \xen}{\rho}} &
\aexpl{by Lemma \ref{lem:esem-merge}} \\
&= \dsem{v}{\esem{\Delta}{\rho}} &
\aexpl{by the induction hypothesis}
\intertext{and for the second part we have}
\esem{\Gamma}{\rho} 
&\preceq \esem{\xen}{\esem{\Gamma}\rho} &
\aexpl{because the $x_1,\ldots,x_n$ are fresh} \\
&\preceq \esem{\Gamma, \xen}\rho &
\aexpl{by Lemma \ref{lem:esem-merge}} \\
&\preceq \esem{\Delta}\rho. &
\aexpl{by the induction hypothesis}
\end{align*}
\end{proof}

\section{Related work}

A large number of developments on formal semantics of functional programming languages in the last two decades build on Launchbury’s work. Many of them implicitly or explicitly rely on the correctness proof as spelled out by Launchbury:

Van Eekelen \& de~Mol \shortcite{mixed} add strictness annotations to the syntax and semantics of Launchbury’s work. They state the correctness as in `Theorem' \ref{thm:false}, without spotting the issue.

Nakata \& Hasegawa \shortcite{nakata} define a small-step semantics for call-by-need and relate it to a Launchbury-derived big-step semantics. They state the correctness with respect to the denotational semantics and reproduce the `Theorem' \ref{thm:false} in the flawed form, in their extended version.

S{\'a}nchez-Gil {\em et~al.} \shortcite{distributed} extend Launchbury's semantics with distributed evaluation.  In their modified natural semantics the heap retains the expression under evaluation with a special flag, marking them as blocked. Furthermore, the expression under evaluation has a name. Thus the non-distributed subset of their semantics is very similar to our stacked semantics and their correctness statement corresponds to Theorem~\ref{thm:thm2}.

An interesting case is the work by Baker-Finch {\em et~al.} on parallel call-by-need: While an earlier report \shortcite{parallel-tr} uses Launchbury’s definitions unmodified and states the flawed `Theorem' \ref{thm:false}, the following publication at ICFP \shortcite{parallel} uses an update-based denotational semantics, unfortunately without motivating that change.

Similarly, Nakata  \shortcite{nakata_blackhole}, who modifies the denotational semantics to distinguish direct cycles from looping recursion, uses update-based semantics without further explanation. 

This list is just a small collection of many more Launchbury-like semantics. Often the relation to a denotational semantics is not stated, but nevertheless they are standing on the foundations laid by Launchbury. Therefore it is not surprising that others have worked on formally fortifying these foundations as well:

S{\'a}nchez-Gil {\em et~al.} identified a step in his adequacy proof relating the standard and the resourced denotational semantics that is not as trivial as it seemed at first and worked out a detailed pen-and-paper proof \shortcite{functionspaces}. They also plan to prove the equivalency between Launchbury’s natural semantics and a variant thereof, which was used by Launchbury in his adequacy proof, in the theorem prover Coq.

As one step in that direction, they address the naming issues and suggest a mixed representation, using de Bruijn indices for locally bound variables and names for free variables \shortcite{nameless}. This corresponds to our treatment of names in the formal development, using the Nominal logic machinery \cite{nominal} locally but not for names bound in heaps.

They also proved that the natural semantics with a modified \sRule{App} that uses an indirection on the heap instead of substitution, can be proven equivalent to the original semantics \cite{indirections}. This is one step towards showing Launchbury's alternative natural semantics equivalent to the original.

\section{Discussion}

Although we have found a flaw in the proof and the formulation of the correctness theorem in Launchbury’s semantics, the essential correctness result (as formulated in Theorem~\ref{thm:main}) stills holds. In that sense our work, especially with the computer-verified proofs in \cite{afp}, actually strengthens the foundations of formal work in that field.

We have provided two ways to fix the problem: One by adding a stack to the operational semantics, and one by changing the denotation of heaps. Both variants have precursors in the literature: The extended heaps of \cite{distributed} resemble the heap-stack-pairs of our semantics, while Baker-Finch {\em et~al.} \shortcite{parallel} and Nakata \shortcite{nakata_blackhole} use right-sided updates in their denotational heap semantics.

Adding the stack to the judgments of the natural semantics, although it does not affect the evaluation of expressions at all, seems to make the semantics more suitable for formal proofs than the original semantics:
\begin{compactitem}
\item Launchbury’s global notion of “fresh variable” is hard to work with, as noted and fixed before by adding to the judgments an explicit list of variable names to avoid \cite{sestoft}. As the stack in our semantics already contains these names, this additional step is not required.
\item While proving Theorem \ref{thm:stackedcorrectness}, the environment $\rho$ does not have to be all-quantified in the inductive step. This simplified the proof a bit and avoids the pitfall that the original proof fell into.
\end{compactitem}
The similarity with the extended heaps in \cite{distributed} further supports the usefulness of the stacked semantics.

We had to spend slightly more pages on the properties of the original denotational semantics (Appendix~\ref{sec:denprops}) than of the update-based denotational semantics (Appendix~\ref{updsemanticsprops}). This is partly due to more details in the proofs in the former section, partly because $+$ behaves nicer than $\sqcup$ (e.g.\ it is defined everywhere). This is orthogonal to the usefulness of the stack-based natural semantics, which could be proven correct with regard to the updates-based denotational semantics as well.

Unless one has to stick with $\sqcup$, e.g. to stay compatible with previous work using this definition, little stands in the way of using the update-based denotational semantics to model the denotation of heaps.

\section{Future work}

Proving correctness is of course only half the battle: The adequacy of Launchbury’s semantics is not yet formally proven. The original paper itself outlines the steps of a proof, and some of these steps have since then been spelled out in greater detail \cite{functionspaces}, and the same authors are currently working on the equivalency of Launchbury’s original and modified natural semantics \cite{indirections}.

We however found that the adequacy is shown easier and more elegantly by taking
care of indirections and blackholing on the denotational side, and have a
complete and machine-checked proof (yet to be published) of that.

We hope that by proving the correctness and adequacy in a theorem prover, we not “just” reinforce our theoretical foundations but also create a practical (for a theoretician’s understanding of practical) tool that can be used to experiment with the many various ways that this semantics can be extended and modified.

\section*{Acknowledgments}

I would like to thank Andreas Lochbihler and Denis Lohner for careful proof-reading and very constructive comments. Furthermore I'd like to thank the anonymous referees at the Journal of Functional Programming for further typo spotting. This work was supported by the Deutsche Telekom Stiftung.

\bibliographystyle{amsalpha}
\bibliography{\jobname}

\appendix

\section{Appendix}
\subsection{Properties of the denotational semantics}
\label{sec:denprops}

This section collects the various technical lemmas about the denotational semantics that we need in the proof for the correctness of the stacked semantics (Theorem~\ref{thm:stackedcorrectness}).

The following two lemmas give the result of looking up a variable $x$ in the denotation of a heap $\esem\Gamma\rho$, depending on whether $x$ is bound in the heap or not. 

\begin{lemma}
\label{lem:esem_this}
If $x \mapsto e \in \Gamma$, then $(\esem{\Gamma}\rho)\, x = \rho\, x \sqcup \dsem{e}{\esem{\Gamma}\rho}$. In particular, $(\esem{\Gamma})\, x = \dsem{e}{\esem{\Gamma}}$.
\end{lemma}

\begin{proof}
by unfolding the fixed point once.
\end{proof}

\begin{lemma}
\label{lem:esem_other}%
\label{lem:remove}
If $x \notin \dom \Gamma$, then $(\esem{\Gamma}\rho)\, x = \rho\, x$.
In particular, if $\dom \Gamma \cap \dom \rho = \emptyset$, then $(\esem{\Gamma}\rho) \setminus \dom\Gamma = \rho$.
\end{lemma}

\begin{proof}
by unfolding the fixed point once.
\end{proof}

The denotation of a heap $\esem\Gamma\rho$ is a refinement of the environment $\rho$, as shown by the next lemma.

\begin{lemma}
$\rho \sqsubseteq \esem{\Gamma}{\rho}$.
\label{lem:rho_below_esem}
\end{lemma}

\begin{proof}
This follows from the fixed-point equation $\esem{\Gamma}{\rho} = \rho \sqcup \dsem{\Gamma}{\esem{\Gamma}{\rho}}$.
\end{proof}

We often need to show that the denotation of a heap is less defined than an environment, and usually do this using the following lemma.
\begin{lemma}
If $\rho \sqsubseteq \rho^*$ and $\dsem{\Gamma}{\rho^*} \sqsubseteq \rho^*$, then $\esem{\Gamma}\rho \sqsubseteq \rho^*$.
\label{lem:esem_below}
\end{lemma}

\begin{proof}
By definition, $\esem{\Gamma}\rho$ is the least fixed point of the functorial $\lambda \rho'. \rho \sqcup \dsem{\Gamma}{\rho'}$, and hence the least pre-fixed point. By assumption, $\rho^*$ is a pre-fixed point, so $\esem{\Gamma}\rho \sqsubseteq \rho^*$ holds.
\end{proof}

The following two lemmas provide a way to replace a binding in a heap.
\begin{lemma}
If $\dsem{e_1}{\esem{x_1 \mapsto e_2, \Gamma}\rho} \sqsubseteq \dsem{e_2}{\esem{x \mapsto e_2, \Gamma}\rho}$
then $\esem{x_1 \mapsto e_1, \Gamma}\rho \sqsubseteq \esem{x_1 \mapsto e_2, \Gamma}\rho$.
\label{lem:esem_subst_expr_below}
\end{lemma}

\begin{proof}
By Lemma \ref{lem:esem_below}, it suffices to show $\rho \sqsubseteq \esem{x \mapsto e_2, \Gamma}\rho$, which follows from  Lemma \ref{lem:rho_below_esem}, and 
\[
(x \mapsto \dsem{e_1}{\esem{x \mapsto e_2, \Gamma}\rho}, \dsem{\Gamma}{\esem{x \mapsto e_2, \Gamma}\rho}) \sqsubseteq \esem{x \mapsto e_2, \Gamma}\rho.
\]
This follows from Lemma \ref{lem:esem_this} for Variables from $\dom\Gamma$, and for $x$ via the assumption, Lemma \ref{lem:esem_this} and transitivity of $\sqsubseteq$.
\end{proof}

\begin{lemma}
If
$\dsem{e_1}{\esem{x \mapsto e_2, \Gamma}\rho} \sqsubseteq \dsem{e_2}{\esem{x \mapsto e_2, \Gamma}\rho}
\text{ and }
\dsem{e_2}{\esem{x \mapsto e_1, \Gamma}\rho} \sqsubseteq \dsem{e_1}{\esem{x \mapsto e_1, \Gamma}\rho}
$
holds then $\esem{x \mapsto e_1, \Gamma}\rho = \esem{x \mapsto e_2, \Gamma}\rho$.
\label{lem:esem_subst_expr}
\end{lemma}

\begin{proof}
By Lemma \ref{lem:esem_subst_expr_below} and antisymmetry of $\sqsubseteq$.
\end{proof}

The next lemmas allows to replace a subexpression $e$ of an expression $e'[e]$ by a variable bound to that subexpression:

\begin{lemma}
Let $z \notin \dom \rho$. Then $\esem{y \mapsto e'[e], z \mapsto e, \Gamma}\rho = \esem{y \mapsto e'[z], z \mapsto e, \Gamma}\rho$.
\label{lem:exp_var_subst}
\label{lem:var_var_subst}
\end{lemma}

\begin{proof}
We have
$\dsem{z}{\esem{y \mapsto e'[e], z \mapsto e, \Gamma}\rho} = \dsem{e}{\esem{y \mapsto e'[e], z \mapsto e, \Gamma}\rho}$ by Lemma \ref{lem:esem_this}, so by the compositionality of the denotational semantics, $\dsem{e'[z]}{{\esem{y \mapsto e'[e], z \mapsto e, \Gamma}\rho}} = \dsem{e'[e]}{{\esem{y \mapsto e'[e], z \mapsto e, \Gamma}\rho}}$ holds.

Analogously, $\dsem{e'[e]}{{\esem{y \mapsto e'[z], z \mapsto e, \Gamma}\rho}} = \dsem{e'[z]}{{\esem{y \mapsto e'[z], z \mapsto e, \Gamma}\rho}}$, so by Lemma \ref{lem:esem_subst_expr} the proof is finished.
\end{proof}

%The next two lemmas are substitution lemmas for the simple cases, i.e. replacing an subexpression by a variable bound to that subexpression, when the subexpression is not inside a binding construct (let or lambda-abstraction).
%\begin{lemma}
%Let $z \notin \dom \rho$. Then $\esem{y \mapsto \sApp e x, z \mapsto e, \Gamma}\rho = \esem{y \mapsto \sApp z x, z \mapsto e, \Gamma}\rho$.
%\label{lem:exp_var_subst}
%\end{lemma}
%
%\begin{proof}
%We have
%$\dsem{z}{\esem{y \mapsto \sApp e x, z \mapsto e, \Gamma}\rho} = \rho\,z \sqcup \dsem{e}{\esem{y \mapsto \sApp e x, z \mapsto e, \Gamma}\rho} = \dsem{e}{\esem{y \mapsto \sApp e x, z \mapsto e, \Gamma}\rho}$, as $z \notin \dom \rho$, so by the denotation of application, $\dsem{\sApp z x}{{\esem{y \mapsto \sApp e x, z \mapsto e, \Gamma}\rho}} = \dsem{\sApp e x}{{\esem{y \mapsto \sApp e x, z \mapsto e, \Gamma}\rho}}$ holds.
%
%Analogously, $\dsem{\sApp e x}{{\esem{y \mapsto \sApp z x, z \mapsto e, \Gamma}\rho}} = \dsem{\sApp z x}{{\esem{y \mapsto \sApp z x, z \mapsto e, \Gamma}\rho}}$, so by Lemma \ref{lem:esem_subst_expr} the proof is finished.
%\end{proof}
%
%\begin{lemma}
%Let $y \notin \dom \rho$. Then $\esem{y \mapsto e, z \mapsto e, \Gamma}\rho = \esem{y \mapsto z, z \mapsto e, \Gamma}\rho$.
%\label{lem:var_var_subst}
%\end{lemma}
%
%\begin{proof}
%Analogously to the proof of Lemma \ref{lem:exp_var_subst}
%\end{proof}

Removing bindings from the denotation of a heap and adding them again does not modify the heap.
\begin{lemma}
$\esem{\Gamma}{(\esem{\Gamma, \Delta} \setminus \dom\Gamma)} = \esem{\Gamma, \Delta}$
\label{lem:redo}
\end{lemma}
\begin{proof}
We use the antisymmetry of $\sqsubseteq$. Note that we use the first inequality in the proof of the second inequality.
\begin{compactitem}[$\sqsubseteq$:]
\item[$\sqsubseteq$:]
By Lemma \ref{lem:esem_below}, it suffices to show
\[
(\esem{\Gamma, \Delta} \setminus \dom\Gamma) \sqcup \dsem{\Gamma}{\esem{\Gamma, \Delta}} \sqsubseteq \esem{\Gamma, \Delta}.
\]
which follows immediately from Lemma \ref{lem:esem_this}.
\item[$\sqsupseteq$:]
Again by Lemma \ref{lem:esem_below}, it suffices to show
\[
\dsem{\Gamma, \Delta}{\esem{\Gamma}{(\esem{\Gamma, \Delta} \setminus \dom\Gamma)}} \sqsubseteq \esem{\Gamma}{(\esem{\Gamma, \Delta} \setminus \dom\Gamma)},
\]
which we verify pointwise. For $x \mapsto e \in \Gamma$, we even have equality:
\begin{align*}
\dsem{e}{\esem{\Gamma}{(\esem{\Gamma, \Delta} \setminus \dom\Gamma)}}
&= (\esem{\Gamma, \Delta} \setminus \dom\Gamma)\, x \sqcup \dsem{e}{\esem{\Gamma}{(\esem{\Gamma, \Delta} \setminus \dom\Gamma)}} \\
\aexpl{because $x\in \dom\Gamma$} \\
&= (\esem{\Gamma}{(\esem{\Gamma, \Delta} \setminus \dom\Gamma)})\, x\\
\aexpl{by Lemma \ref{lem:esem_this}.} \\
\end{align*}
For $x \mapsto e \in \Delta$, we have 
\begin{align*}
\dsem{e}{\esem{\Gamma}{(\esem{\Gamma, \Delta} \setminus \dom\Gamma)}}
&\sqsubseteq \dsem{e}{\esem{\Gamma, \Delta}} \\
\aexpl{by case $\sqsubseteq$ and the monotonicity of $\dsem{e}\_$.} \\
&= (\esem{\Gamma, \Delta})\,x. \\
\aexpl{by Lemma \ref{lem:esem_this}} \\
&= (\esem{\Gamma, \Delta}\setminus \dom\Gamma)\,x. \\
\aexpl{because $x\notin \dom\Gamma$.} \\
&= (\esem{\Gamma}{(\esem{\Gamma, \Delta} \setminus \dom\Gamma)})\,x \\
\aexpl{by Lemma \ref{lem:esem_other}, as $x\notin \dom\Gamma$.}
\end{align*}
\end{compactitem}
\end{proof}

Fresh variables do not affect the denotation of expressions, as shown in the next tree lemmas.
\begin{lemma}
If all variables in $S$ are fresh with regard to $e$, then $\dsem{e}\rho = \dsem{e}{\rho \setminus S}$.
\label{lem:see_through_fresh}
\end{lemma}

\begin{proof}
by induction on $e$.
\end{proof}

\begin{lemma}
Let $x$ be fresh. Then $\esem \Gamma \rho = (\esem {x \mapsto e, \Gamma} \rho) \setminus \{x\}$.
\label{lem:addvar}
\end{lemma}

\begin{proof}
If $x$ is fresh with regard to $\Gamma$, then a binding of $x$ in $\rho$ does not affect $\dsem \Gamma\rho$, as shown by induction on the expressions bound in $\Gamma$.
\end{proof}

\begin{lemma}
If $\dom\Gamma$ is fresh with regard to $\Delta$ and $\rho$, then $\esem{\Gamma}{\esem{\Delta}\rho} = \esem{\Gamma, \Delta}\rho$.
\label{lem:esem_merge}
\end{lemma}

\begin{proof}
We show this by using antisymmetry.
\begin{compactitem}[$\sqsubseteq$:]
\item[$\sqsubseteq$:]
By invoking Lemma \ref{lem:esem_below} twice, it suffices to show 
\begin{compactitem}
\item $\rho \sqsubseteq \esem{\Gamma,\Delta}\rho$, which follows from Lemma \ref{lem:rho_below_esem}, as well as 
\item $\dsem{\Delta}{\esem{\Gamma,\Delta}\rho} \sqsubseteq \esem{\Gamma,\Delta}\rho$ and
\item $\dsem{\Gamma}{\esem{\Gamma,\Delta}\rho} \sqsubseteq \esem{\Gamma,\Delta}\rho$, which follows from Lemma \ref{lem:esem_this}.
\end{compactitem}
\item[$\sqsupseteq$:] By Lemma 10, it suffices to show $\rho \sqsubseteq \esem{\Gamma}{\esem{\Delta}\rho}$, for which we invoke Lemma \ref{lem:rho_below_esem} twice, and $\dsem{\Gamma,\Delta}{\esem{\Gamma}{\esem{\Delta}\rho}} \sqsubseteq \esem{\Gamma}{\esem{\Delta}\rho}$. For the latter we consider two cases:
\begin{compactenum}
\item For $x \mapsto e \in \dom\Gamma$, this follows from 
Lemma \ref{lem:esem_this}.
\item For $x \mapsto e \in \dom\Delta$, we have that $\dom\Gamma$ is fresh with regard to $e$, so
\begin{align*}
\dsem{e}{\esem{\Gamma}{\esem{\Delta}\rho}}
&= \dsem{e}{\esem{\Delta}\rho}&
\aexpl{by Lemma \ref{lem:see_through_fresh}} \\
&= (\esem{\Delta}\rho)\, x&
\aexpl{by Lemma \ref{lem:esem_this}} \\
&= (\esem{\Gamma}{\esem{\Delta}\rho})\, x&
\aexpl{by Lemma \ref{lem:esem_other} and $x \notin \dom\Gamma$.}
\end{align*}
\end{compactenum}

\end{compactitem}
\end{proof}

The last lemma of this section states that unpacking a let-expression on the heap preserves the denotation of the existing bindings.

\begin{lemma}
$\esem{z \mapsto \sLet{\xeng}e, \Gamma} \preceq \esem{\xen, z \mapsto e, \Gamma}$.
\label{lem:let_unfold}
\end{lemma}
\begin{proof}
Let $\Gamma' \Coloneqq (\xen, z \mapsto e, \Gamma)$ and $e' \Coloneqq \sLet{\xeng}e$. The lemma follows from $\esem{z \mapsto e', \Gamma} = \esem{\Gamma'} \setminus \{x_1,\ldots,x_n\}$, which we show using antisymmetry.
\begin{compactitem}[$\sqsubseteq$:]
\item[$\sqsubseteq$:]
The left hand side is a least fixed point, so it suffices to show $\dsem{z \mapsto e', \Gamma}{\Gamma' \setminus \{x_1,\ldots,x_n\}} = \Gamma' \setminus \{x_1,\ldots,x_n\}$. For variables in the domain of $\Gamma$, this follows from Lemma \ref{lem:esem_this}. For $z$, we have
\begin{align*}
\dsem{e'}{\Gamma' \setminus \{x_1,\ldots,x_n\}}
&= \dsem{e}{\esem{x_1 \mapsto e_1,\ldots,x_n\mapsto e_n}{(\esem{\Gamma'} \setminus \{x_1,\ldots,x_n\})}} \\
\aexpl{by the denotation of let-expressions} \\
&= \dsem{e}{\esem{\Gamma'}} \\
\aexpl{by Lemma \ref{lem:redo}} \\
&\sqsubseteq \Gamma'\, z. \\
\aexpl{by Lemma \ref{lem:esem_this}.}
\end{align*}

\item[$\sqsupseteq$:]
First note that
\begin{align*}
\esem{\Gamma'} \sqsubseteq \esem{\xen}{\esem{z \mapsto e', \Gamma}} \tag{$\ast$}
\end{align*}
for which it suffices to show
\[
\dsem {\Gamma'}{\esem{\xen}{\esem{z \mapsto e', \Gamma}}} \sqsubseteq \esem{\xen}{\esem{z \mapsto e', \Gamma}},
\]
which falls into two cases:
\begin{compactenum}
\item For $z$, we have 
\begin{align*}
\dsem{e}{\esem{\xen}{\esem{z \mapsto e', \Gamma}}}
&= \dsem{e'}{\esem{z \mapsto e', \Gamma}} \\
\aexpl{by the denotation of let-expressions} \\
&= (\esem{z \mapsto e', \Gamma})\, z\\
\aexpl{by Lemma \ref{lem:esem_this}} \\
&= (\esem{\xen}{\esem{z \mapsto e', \Gamma}})\, z \\
\aexpl{by Lemma \ref{lem:esem_other}}
\end{align*}

\item For $x \mapsto e^* \in (\xen, \Gamma)$, we have
\begin{align*}
\dsem{e^*}{\esem{\xen, z \mapsto e', \Gamma}}
&= (\esem{\xen, z \mapsto e', \Gamma})\, x
\end{align*}
by Lemma \ref{lem:esem_this}. Using Lemma \ref{lem:esem_merge} this concludes case 2, as the $x_1,\ldots,x_n$ are fresh with regard to $(z\mapsto e', \Gamma)$.
\end{compactenum}

Now we can show
\begin{align*}
\esem{\Gamma'} \setminus \{x_1,\ldots,x_n\}
&\sqsubseteq  \esem{\xen}{\esem{z \mapsto e', \Gamma}}  \setminus \{x_1,\ldots,x_n\} \\
\aexpl{by ($\ast$)} \\
& = {\esem{z \mapsto e', \Gamma}} \\
\aexpl{by Lemma \ref{lem:remove}.}
\end{align*}
\end{compactitem}

\end{proof}

\subsection{Properties of the update-based denotational semantics}
\label{updsemanticsprops}

To reproduce Launchbury's correctness proof (Theorem~\ref{thm:thm2}) with regard to the update-based semantics we first show some lemmas about the denotational semantics.
 
% In order to prove correctness of Launchbury's natural semantics with regard to this denotational semantics, using the prove as given in \cite{launchbury}, we need to state some lemmas about the denotational semantics. % Again we recommend to skip this section at first and follow the references as needed.

The lemma that justifies the introduction of the update-based semantics is the following, which is the equality that was used in the original proof but does not hold for the standard denotational semantics (see Section~\ref{counterexample}).

\begin{lemma}
\label{lem:esemu_this}
For $(x\mapsto e)\in \Gamma$ we have $(\esemu{\Gamma}\rho)\,x = \dsem{e}{\esemu{\Gamma}\rho}$.
\end{lemma}

\begin{proof}
by unrolling the fixed point once.
\end{proof}

The other case when looking up a variable in the denotation of a heap is

\begin{lemma}
\label{lem:esemu_other}
For $x \notin \dom\Gamma$ we have $(\esemu{\Gamma}\rho)\,x = \rho\, x$.
\end{lemma}

\begin{proof}
by unrolling the fixed point once.
\end{proof}

We show an alternative, iterative definition of the heap semantics.

\begin{lemma}
$
\esemu{x \mapsto e, \Gamma}\rho = \big(\mu \rho'.\,  \rho + (\esemu{\Gamma}{\rho'})|_{\dom\Gamma} + (x \mapsto \dsem{e}{\esemu{\Gamma}\rho'})\big).
$
\label{lem:iter}
\end{lemma}

A corresponding lemma can be found in Launchbury \shortcite{launchbury}, but without proof. As the proof involves some delicate fixed-point-juggling, we include it here in detail:

\begin{proof}
Let $L = (\lambda \rho'.\, \rho + \dsem{x\mapsto e, \Gamma}{\rho'})$ be the functorial of the fixed point on the left hand side, $R$~be the functorial on the right hand side.

By Lemmas \ref{lem:esemu_this} and \ref{lem:esemu_other}, we have
\begin{compactenum}[(1)]
\item $(\mu L)\, y = \dsem{e'}{\mu L}$ for $y\mapsto e'\in\dom\Gamma$,
\item $(\mu L)\, x = \dsem{e}{\mu L}$,
\item $(\mu L)\, y = \rho\, y$ for $y \notin \{x\}\cup\dom\Gamma$.
\end{compactenum}
Similarly, by unrolling the fixed points, we have
\begin{compactenum}[(1)]
\item[(4)] $(\mu R)\, y = \dsem{e'}{\esemu{\Gamma}{(\mu R)}}$ for $y \mapsto e'\in\dom\Gamma$,
\item[(4)] $(\mu R)\, x = \dsem{e}{\esemu{\Gamma}{(\mu R)}}$,
\item[(4)] $(\mu R)\, y = \rho\, y$ for $y \notin \{x\}\cup\dom\Gamma$,
\end{compactenum}
and also for $\rho' \in \sEnv$ (in particular for $\rho' = (\mu L)$, $(\mu R)$), again using Lemmas \ref{lem:esemu_this} and \ref{lem:esemu_other},
\begin{compactenum}[(1)]
\item[(7)] $(\esemu{\Gamma}{\rho'})\,y = \dsem{e}{\esemu{\Gamma}{\rho'}}$ for $y \mapsto e' \in \dom\Gamma$,
\item[(8)] $(\esemu{\Gamma}{\rho'})\,y = \rho'\, y$ for $y \notin \dom\Gamma$.
\end{compactenum}

\medskip
We obtain
\begin{compactenum}[(1)]
\item[(9)] $\esemu{\Gamma}{(\mu R)} = (\mu R)$
\end{compactenum}
from comparing (4)--(6) with (7) and (8). We can also show
\begin{compactenum}[(1)]
\item[(10)] $\esemu{\Gamma}{(\mu L)} = (\mu L)$,
\end{compactenum}
by antisymmetry and use that least fixed points are least pre-fixed points:
\begin{compactitem}[$\sqsubseteq$:]
\item[$\sqsubseteq$:] We need to show that $(\mu L) + \dsem{\Gamma}{(\mu L)} \sqsubseteq (\mu L)$, which follows from (1). 
\item[$\sqsupseteq$:] We need to show that $\esemu{\Gamma}{(\mu L)} + \dsem{x\mapsto e, \Gamma}{\esemu{\Gamma}{(\mu L)}} \sqsubseteq \esemu{\Gamma}{(\mu L)}$. For $\dom\Gamma$, this follows from (7), so we show $\dsem{e}{\esemu{\Gamma}{(\mu L)}} \sqsubseteq (\mu L)\, x = \dsem{e}{(\mu L)}$, which follows from the monotonicity of $\dsem{e}{\_}$ and case $\sqsubseteq$.
\end{compactitem}

To show the lemma, $(\mu L) = (\mu R)$, we use the antisymmetry of $\sqsubseteq$ and the leastness of least fixed points:
\begin{compactitem}[$\sqsubseteq$:]
\item[$\sqsubseteq$:] We need to show that $L\, (\mu R) = \mu R$, i.e.
\begin{compactitem}
\item $\rho\,y = (\mu R)\, y$ for $y \notin \{x\}\cup \dom\Gamma$, which follows from (6),
\item $\dsem{e'}{\mu R} = (\mu R)\, y$ for $y \mapsto e' \in \Gamma$, which follows from (4) and (9) and
\item $\dsem{e}{\mu R} = (\mu R)\, x$, which follows from (5) and (9).
\end{compactitem}
\item[$\sqsupseteq$:] Now we have to show that $R\ (\mu L) = (\mu L)$, i.e.
\begin{compactitem}
\item $\rho\,y = (\mu L)\, y$ for $y \notin \{x\}\cup \dom\Gamma$, which follows from (3),
\item $\dsem{e'}{\esemu{\Gamma}{(\mu L)}} = (\mu L)\, y$ for $y \mapsto e' \in \Gamma$, which follows from (1) and (10), and
\item $\dsem{e}{\esemu{\Gamma}{(\mu L)}} = (\mu L)\, x$, which follows from (2) and (10).
\end{compactitem}
\end{compactitem}
\end{proof}

Next we prove that substitutions in terms are equivalent to indirections on the heap. This lemma will also be useful when proving adequacy: Instead of bringing the operational semantics closer to the denotational semantics by replacing the substiution in the operational semantics with an indirection via the heap, as proposed in \cite{launchbury} and carried out in \cite{indirections}, this lemma performs the step on the denotational side.

\begin{lemma}
If $y$ is fresh with regard to $\rho$, then
\label{lem:subst}
\[
\dsem{e}{\rho (y \mapsto \dsem{x}\rho)} = \dsem{ e[x/y]}{\rho}.
\]
\end{lemma}

\begin{proof}
By induction on $e$. All cases are trivial but case \sRule{Let}, which requires some shuffling of fixed points.
We need to show that
\[
\dsem{\sLet{\xeng}e}{\rho (y \mapsto \dsem{x}{\rho})} = \dsem{ (\sLet{\xeng}e)[x/y]}{\rho}.
\]
using
$\dsem{e}{\rho' (y \mapsto \dsem{x}{\rho'})} = \dsem{ e[x/y]}{\rho'}$ and
$\dsem{e_i}{\rho' (y \mapsto \dsem{x}{\rho'})} = \dsem{ e_i[x/y]}{\rho'}$ for $i=1,\ldots,n$.

Let $\Gamma = (\xen)$. The variables $x_1,\ldots,x_n$ are fresh. In particular, none of them are $x$ or $y$.

We first show that
\begin{align*}
\esemu{\Gamma}{(\rho (y \mapsto \dsem{x}{\rho}))} = (\esemu{\Gamma[x/y]}{\rho})(y \mapsto \dsem{x}{\esemu{\Gamma[x/y]}{\rho}}) \tag{$\ast$}
\end{align*}
using the antisymmetry of $\sqsubseteq$. Let $\rho_L$ and $\rho_R$ denote the left- and right-hand-side of the equation.
\begin{compactitem}[$\sqsubseteq$:]
\item[$\sqsubseteq$:]
By the leastness of the fixed point, it suffices to show that $\rho (y \mapsto \dsem{x}{\rho}) + \dsem{\Gamma}{\rho_R} = \rho_R$, which we verify pointwise.
\begin{compactitem}
\item For $y$, because the $x_i$ are fresh, we have $\rho\, x$ on both sides.
\item For the $x_i$, we have to show $\dsem{e_i}{\rho_R} = \dsem{e_i[x/y]}{\esemu{\Gamma[x/y]}{\rho}}$, which is our induction hypothesis (with $\rho' = \rho_R$).
\item For any other variable $x'$, we have $\rho\, x'$ on both sides.
\end{compactitem}

\item[$\sqsupseteq$:]
Clearly $\dsem{x}{\esemu{\Gamma[x/y]}\rho} = \rho\, x = \rho_L\, y$, so it remains to show that $\esemu{\Gamma[x/y]}{\rho} \sqsubseteq \rho_L \setminus \{y\}$. We again use the leastness of the fixed point and verify the inequality
$\rho + \dsem{\Gamma[x/y]}{\rho_L\setminus\{y\}} \sqsubseteq \rho_L\setminus\{y\}$ pointwise:
\begin{compactitem}
\item For $y$, as $y$ is fresh with regard to $\rho$, we have $\bot$ on both sides.
\item For the $x_i$, we have to show $\dsem{e_i[x/y]}{\rho_L \setminus \{y\}} = \dsem{e_i}{\rho_L}$, which follows from our induction hypothesis (with $\rho' = \rho_L \setminus\{y\})$ and $(\rho_L \setminus\{y\})(y \mapsto \dsem{x}{\rho_L\setminus\{y\}}) = \rho_L$.
\item For any other variable $x'$, we have $\rho\, x'$ on both sides.
\end{compactitem}
\end{compactitem}

Finally, we calculate
\begin{align*}
\beginright
\dsem{\sLet{\xeng}e}{\rho (y \mapsto \dsem{x}{\rho})} \\
&= \dsem{e}{\rho_L} &
\aexpl{by the denotation of let expressions} \\
&= \dsem{e}{\rho_R} &
\aexpl{by $(\ast)$} \\
&= \dsem{e[x/y]}{\esemu{\Gamma[x/y]}{\rho}} &
\aexpl{by the induction hypothesis}  \\
&= \dsem{ (\sLet{\xeng}e)[x/y]}{\rho} &
\aexpl{by the denotation of let expressions.}
\end{align*}
\end{proof}

The final lemma required for the correctness proof of shows that the denotation of a heap with only fresh variables can be merged with the heap it was defined over:

\begin{lemma}
\label{lem:esem-merge}
If $\dom \Gamma$ is fresh with regard to $\Delta$ and $\rho$, then
\[
\esem{\Gamma}{\esem{\Delta}\rho} = \esem{\Gamma, \Delta}\rho.
\]
\end{lemma}

\begin{proof}
First note that 
\begin{align*}
\esem{\Delta}\rho \preceq \esem{\Gamma}{\esem{\Delta}\rho}, \tag{$\ast$}
\end{align*}
as the variables bound in $\Gamma$ are fresh and existing bindings in $\esem{\Delta}\rho$ keep their semantics.

We use the antisymmetry of $\sqsubseteq$, and the leastness of least fixed points.
\begin{compactitem}[$\sqsubseteq$:]
\item[$\sqsubseteq$:] We need to show that $\esem{\Delta}\rho + \dsem\Gamma{\esem{\Delta,\Gamma}\rho} = \esem{\Delta,\Gamma}\rho$. This follows from ($\ast$) and from unrolling the fixed point on the right hand side once.
\item[$\sqsupseteq$:] We need to show that $\rho + \dsem{\Gamma, \Delta}{\esem{\Gamma}{\esem{\Delta}\rho}} = \esem{\Gamma}{\esem{\Delta}\rho}$, which we verify pointwise.
\begin{compactitem}
\item For $x\in \dom\Gamma$, this follows from unrolling the fixed point on the right hand side once.
\item For $x\mapsto e \in \dom\Delta$ (and hence $x\notin \dom \Gamma$), we have
\begin{align*}
(\rho + \dsem{\Gamma, \Delta}{\esem{\Gamma}{\esem{\Delta}\rho}})\, x 
&= \dsem{e}{\esem{\Gamma}{\esem{\Delta}\rho}} \\
&= \dsem{e}{\esem{\Delta}\rho} \\
\aexpl{because $\dom\Gamma$ is fresh with regard to $e$} \\
&= (\esem\Delta\rho)\, x\\
\aexpl{by unrolling the fixed point} \\
&= (\dsem{\Gamma}{\esem\Delta\rho})\, x\\
\aexpl{because $x\notin \dom\Gamma$}
\end{align*}
\item For $x\notin \dom \Gamma \cup \dom \Delta$, we have $\rho\, x$ on both sides.
\end{compactitem}
\end{compactitem}
\end{proof}

\end{document}